\documentclass[a4paper]{llncs}

\pdfoutput=1 

\usepackage{amsfonts, amsmath, amssymb, alltt, color, pbox, enumitem}
\usepackage{algorithm}
\usepackage{url}
\usepackage{algcompatible}
\allowdisplaybreaks

\urldef{\mailsa}\path|{flavio.ferrarotti, klaus-dieter.schewe, loredana.tec}@scch.at|
\urldef{\mailsb}\path|qing.wang@anu.edu.au|

\begin{document}

\mainmatter  % start of an individual contribution

\title{A Logic for Non-Deterministic Parallel Abstract State Machines\thanks{Work supported by the {\bf Austrian Science Fund (FWF: [P26452-N15])}. Project: \emph{Behavioural Theory and Logics for Distributed Adaptive Systems}. The final publication is available at Springer via \url{http://dx.doi.org/10.1007/978-3-319-30024-5_18}}}

\author{Flavio Ferrarotti\inst{1}
\and Klaus-Dieter Schewe\inst{1}  \and Loredana Tec\inst{1}  \and Qing Wang\inst{2}}
%

\institute{Software Competence Center Hagenberg, A-4232 Hagenberg, Austria\\
\mailsa\\
\and
Research School of Computer Science, The Australian National University\\
\mailsb\\
%\url{http://www.springer.com/lncs}}
}

\toctitle{Lecture Notes in Computer Science}
\tocauthor{Authors' Instructions}
\maketitle

\vspace{-0.4cm}
\begin{abstract}
We develop a logic which enables reasoning about single steps of non-deterministic parallel Abstract State Machines (ASMs). Our logic builds upon the unifying logic introduced by Nanchen and St\"ark for reasoning about hierarchical (parallel) ASMs. Our main contribution to this regard is the handling of non-determinism (both bounded and unbounded) within the logical formalism. Moreover, we do this without sacrificing the completeness of the logic for statements about single steps of non-deterministic parallel ASMs, such as invariants of rules, consistency conditions for rules, or step-by-step equivalence of rules.
\end{abstract}

\section{Introduction}
Gurevich's Abstract State Machines (ASMs) provide not only a formal theory of algorithms, but also are the basis for a general software engineering method based in the specification of higher-level ground models and step-by-step refinement. 
%This has lead to many success stories including the full formalization of several programming languages such as Prolog, Java and C\#, the formalization of business process models such as BPMN, the investigation of formal properties of stream queries, and the modelling of service oriented architectures, in particular Web services. 
Chapter 9 in the book~\cite{boerger:2003} gives a summary of many application projects that have developed complex systems solutions on the grounds of ASMs. A major advantage of the ASM method and a key for its success resides in the fact that it provides, not only a simple and precise framework to communicate and document design ideas, but also an accurate and checkable overall understanding of complex systems. In this context, formal verification of dynamic properties for given ASMs is a fundamentally important task, in particular in the case of modelling safety critical systems, where there is a need to ensure the integrity and reliability of the system. Clearly, a logical calculi appropriate for the formalisation and reasoning about dynamic properties of ASMs is an essential and valuable tool for this endeavour.

Numerous logics have been developed to deal with specific features of ASM verification such as correctness and deadlock-freeness (see Section~9.4.3 in the book \cite{boerger:2003}) for detailed references), but a complete logic for ASMs was only developed in~\cite{RobertLogicASM} by Nanchen and St\"ark. The logic formalizes properties of a single step of an ASM, which permits to define Hilbert-style proof theory and to show its completeness. In this work the treatment of non-determinism was deliberately left out. Same as parallelism, which is on the other hand captured by the logic for ASMs of Nanchen and St\"ark, non-determinism is also a prevalent concept in the design and implementation of software systems, and consequently a constitutive part of the ASM method for systems development~\cite{boerger:2003}.
Indeed, nondeterminism arises in the specification of many well known algorithms and software applications. Examples range from graph algorithms, such as minimum spanning tree and shortest path, to search techniques whose objective is to arrive at some admissible goal state (as in the n-queens and combinatorial-assignment problems \cite{Floyd67}), and learning strategies such as converging on some classifier that labels all data instances correctly \cite{Vapnik95}.  Non-deterministic behavior is also common in cutting edge fields of software systems. Distributed systems frequently need to address non-deterministic behaviour such as changing role (if possible) as strategic response to observed problems concerning load, input, throughput, etc. Also, many cyber-physical systems and hybrid systems such as railway transportation control systems  \cite{Alur2015} and  systems used in high-confidence medical healthcare devices exhibit highly non-deterministic behaviour.

Notice that although we could say that there is a kind of latent parallelism in non-determinism, they represent completely different behaviours and thus both are needed to faithfully model the behaviour of complex systems, more so in the case of the ASM method where the ability to model systems at every level of abstraction is one of its main defining features. For instance, while a nondeterministic action can evaluate to multiple behaviors, only if at least one of these behaviors does not conflict with concurrent tasks, then there is an admissible execution of the action in parallel with these tasks.  

The ASM method allows for two different, but complementary, approaches to non-determinism. The first approach assumes that choices are made by the environment via monitored functions that can be viewed as external oracles. In this case, non-deterministic ASMs are just interactive ASMs. The second approach assumes the ASMs themselves rather than the environment, to have the power of making non-deterministic choices. In this case the one-step transition function of the ASMs is no longer a function but a binary relation. This is also the approach followed by non-deterministic Turing machines. However, in the case of non-deterministic Turing machines the choice is always bounded by the transition relation. For ASMs the non-determinism can also be unbounded, i.e., we can choose among an infinite number of possibilities. Clearly, unbounded non-determinism should also be allowed if we want our ASMs to be able to faithfully model algorithms at any level of abstraction. 

In this work we develop a logic which enables reasoning about single steps of non-deterministic parallel ASMs, i.e., ASMs which include the well known \textbf{choose} and \textbf{forall} rules~\cite{boerger:2003}. This builds upon the complete logic introduced in the work of Nanchen and St\"ark~\cite{RobertLogicASM} for reasoning about single steps of hierarchical ASMs. Hierarchical ASMs capture the class of synchronous and deterministic parallel algorithms in the precise sense of the ASM thesis of Blass and Gurevich~\cite{blass:tocl2003,GurevichParallelCorrection08} (see also \cite{FerrarottiSTW15}). Our main contribution to this regard is the handling of non-determinism (both bounded and unbounded) within the logical formalism. More importantly, this is done without sacrificing the completeness of the logic. As highlighted by Nanchen and St\"ark~\cite{RobertLogicASM}, non-deterministic transitions manifest themselves as a difficult task in the logical formalisation for ASMs.

The paper is organized as follows. The next section introduces the required background from ASMs. Section~\ref{non-det-asms} formalises the model of non-deterministic parallel ASM used through this work. In Section~\ref{sec:TheLogic} we introduce the syntax and semantics of the proposed logic for non-deterministic parallel ASMs. Section~\ref{sec:proof_system} presents a detailed discussion regarding consistency and update sets, and the formalisation of a proof system. In Section~\ref{sec:Derivation} we use the proof system to derive some interesting properties of our logic, including known properties of the ASM logic in \cite{RobertLogicASM}. In Section~\ref{sec:completeness} we present our main result, namely that the proposed logic is complete for statements about single steps of non-deterministic parallel ASMs, such as invariants of rules, consistency conditions for rules, or step-by-step equivalence of rules. We conclude our work in Section~\ref{conclusions}.

\section{Preliminaries}

The concept of Abstract State Machines (ASMs) is well known~\cite{boerger:2003}. In its simplest form an ASM is a finite set of so-called \emph{transition rules} of the form \textbf{if} \emph{Condition} \textbf{then} \textit{Updates} \textbf{endif} which transforms abstract states. The condition or guard under which a rule is applied is an arbitrary first-order logic sentence. \emph{Updates} is a finite set of assignments of the form $f(t_1, \ldots, t_n) := t_0$ which are executed in parallel. The execution of $f(t_1, \ldots, t_n) := t_0$ in a given state proceeds as follows: first all parameters $t_0, t_1, \ldots t_n$ are evaluated to their values, say $a_0, a_1, \ldots, a_n$, then the value of $f(a_1, \ldots, a_n)$ is updated to $a_0$, which represents the value of $f(a_1, \ldots, a_n)$ in the next state. Such pairs of a function name $f$, which is fixed by the signature, and optional argument $(a_1, \ldots, a_n)$ of dynamic parameters values $a_i$, are called \emph{locations}. They represent the abstract ASM concept of memory units which abstracts from particular memory addressing. Location value pairs $(\ell,a)$, where $\ell$ is a location and $a$ a value, are called \emph{updates} and represent the basic units of state change.       

The notion of ASM \emph{state} is the classical notion of \emph{first-order structure} in mathematical logic. For the evaluation of first-order terms and formulae in an ASM state, the standard interpretation of function symbols by the corresponding functions in that state is used. As usually in this setting and w.l.o.g., we treat predicates as characteristic functions and constants as $0$-ary functions.    

The notion of the ASM \emph{run} is an instance of the classical notion of the computation of transition systems. An ASM computation step in a given state consists in executing \emph{simultaneously} all updates of all transition rules whose guard is true in the state, if these updates are consistent, in which case the result of their execution yields a next state. In the case of inconsistency, the computation does not yield a next state. A set of updates is \emph{consistent} if it contains no pairs $(\ell, a)$, $(\ell, b)$ of updates to a same location $\ell$ with $a \neq b$.  

Simultaneous execution, as obtained in one step through the execution of a set of updates, provides a useful instrument for high-level design to locally describe a global state change. This synchronous parallelism is further enhanced by the transition rule \textbf{forall} $x$ \textbf{with} $\varphi$ \textbf{do} $r$ \textbf{enddo} which expresses the simultaneous execution of a rule $r$ for each $x$ satisfying a given condition $\varphi$. 

Similarly, non-determinism as a convenient way of abstracting from details of scheduling of rule executions can be expressed by the rule \textbf{choose} $x$ \textbf{with} $\varphi$ \textbf{do} $r$ \textbf{enddo}, which means that $r$ should be executed with an arbitrary $x$ chosen among those satisfying the property $\varphi$. 

The following example borrowed from~\cite{boerger:2003} clearly illustrates the power of the \textbf{choose} and \textbf{forall} rules. 

\begin{example}\label{ex-diff-words}
The following ASM generates all and only the pairs $vw \in A^*$ of different words $v, w$ of same length (i.e., $v \neq w$ and $|v|=|w|$).\\[-0.8cm]
\begin{alltt}
\small
choose \(n,i\) with \(i<n\) do
   choose \(a,b\) with \({{a}\in{A}}\wedge{{b}\in{A}}\wedge{{a}\neq{b}}\) do
      \({v(i)}:={a}\)
      \({w(i)}:={b}\)
      forall \(j\) with \({j<n}\wedge{{j}\neq{i}}\) do
         choose \(a,b\) with \({{a}\in{A}}\wedge{{b}\in{A}}\) do
            \({v(j)}:={a}\)
            \({w(j)}:={b}\)
         enddo
      enddo
   enddo
enddo
\end{alltt}
When all possible choices are realized, the set of reachable states of this ASM is the set of all ``$vw$'' states with $v \neq w$ and $|v| = |w|$. 
\end{example}

\section{Non-Deterministic Parallel ASMs}\label{non-det-asms}

It is key for the completeness of our logic to make sure that the ASMs do not produce infinite update sets. For that we formally define ASM states as simple metafinite structures~\cite{graedel:infcomp1998} instead of classical first-order structures, and restrict the variables in the \textbf{forall} rules to range over the finite part of such metafinite states. Nevertheless, the class of algorithms that are captured by these ASM machines coincides with the class of parallel algorithms that satisfy the postulates of the parallel ASM thesis of Blass and Gurevich~\cite{blass:tocl2003,GurevichParallelCorrection08} (see \cite{FerrarottiSTW15} for details). 

A \emph{metafinite structure} $S$ consists of: a finite first-order structure $S_1$ --the \emph{primary part} of $S$; a possibly infinite first-order structure $S_2$ --the \emph{secondary part} of $S$; and a finite set of functions which map elements of $S_1$ to elements of $S_2$ --the \emph{bridge functions}. A signature $\Upsilon$ of metafinite structures comprises a sub-signature $\Upsilon_1$ for the primary part, a sub-signature $\Upsilon_2$ for the secondary part and a finite set
$\mathcal{F}_b$ of bridge function names. The \emph{base set} of a
state $S$ is a nonempty set of values $B=B_1\cup B_2$, where
$B_1$ is the finite domain of $S_1$, and $B_2$ is the possibly infinite domain of $S_2$. Function
symbols $f$ in $\Upsilon_1$ and $\Upsilon_2$ are
interpreted as functions $f^S$ over $B_1$ and $B_2$, respectively. The
 interpretation of a n-ary function symbol $f\in\mathcal{F}_b$
defines a function $f^S$ from $B^{n}_1$ to $B_2$.
As usual, we distinguish between \emph{updatable} dynamic functions and static
functions.

%Let $S$ be a state over $\Upsilon$, $f\in\Upsilon$ be a dynamic function symbol of arity $n$ and $a_1,...,a_n$ be elements in $B_1$ %or $B_2$ depending on whether $f \in \Upsilon_{1} \cup {\cal F}_b$ or $f \in \Upsilon_2$, respectively.
%Then $(f,(a_1,...,a_n))$ is called a {\em location} of $S$. An
%\emph{update} of $S$ is a pair $(\ell,b)$, where $\ell$ is a
%location and $b$ (which belong to $B_1$ or $B_2$ depending on whether $f \in \Upsilon_1$ or $f \in \Upsilon_2 \cup {\cal F}_b$, resp%ectively) is the \emph{update value} of $\ell$. To simplify notation we write $(f
%,(a_1,\dots,a_n),b)$ for the update $(\ell,b)$ with the location $\ell = (f,(a_1,\dots,a_n))$. The
%interpretation of $\ell$ in $S$ is called the \emph{content} of
%$\ell$ in $S$, denoted by $val_{S}(\ell)$. An \emph{update set}
%$\Delta$ is a set of updates.

Let $\Upsilon=\Upsilon_1\cup\Upsilon_2\cup\mathcal{F}_b$ be a signature of metafinite states. Fix a countable set $\mathcal{X} = \mathcal{X}_{1} \cup \mathcal{X}_2$
of first-order variables. Variables in $\mathcal{X}_{1}$, denoted with standard lowercase letters $x, y, z , \ldots$, range over the primary part of a meta-finite state (i.e., the finite set $B_{1}$), whereas variables in $\mathcal{X}_2$, denoted with typewriter-style lowercase letters $\texttt{x}, \texttt{y}, \texttt{z}, \ldots$, range over $B_2$.
The set of first-order terms ${\cal T}_{\Upsilon, {\cal X}}$ of vocabulary $\Upsilon$ is defined in a similar way than in meta-finite model theory \cite{graedel:infcomp1998}. That is, ${\cal T}_{\Upsilon, {\cal X}}$ is constituted by the set $\mathcal{T}_{p}$ of \emph{point terms} and the set $\mathcal{T}_{a}$ of \emph{algorithmic terms}. The set of point terms $\mathcal{T}_{p}$ is the closure of the set ${\cal X}_1$ of variables under the application of function symbols in $\Upsilon_1$. The set of algorithmic terms $\mathcal{T}_{a}$ is defined inductively: Every variable in ${\cal X}_2$ is an algorithmic term in $\mathcal{T}_{a}$; If $t_1, \ldots, t_n$ are point terms in $\mathcal{T}_{p}$ and $f$ is an $n$-ary bridge function symbol in $\mathcal{F}_b$, then $f(t_1, \ldots, t_n)$ is an algorithmic term in $\mathcal{T}_{a}$; if $t_1, \ldots, t_n$ are algorithmic terms in $\mathcal{T}_{a}$ and $f$ is an $n$-ary function symbol in $\Upsilon_2$, then $f(t_1, \ldots, t_n)$ is an algorithmic term in $\mathcal{T}_{a}$; nothing else is an algorithmic term in $\mathcal{T}_{b}$.

%Let $S$ be a meta finite state of signature $\Upsilon$. A \emph{valuation} or \emph{variable assignment} $\zeta$ is a function that assigns to every variable in ${\cal X}_1$ a value in the base set $B_1$ of the primary part of $S$ and to every variable in ${\cal X}_2$ a value in the base set $B_2$ of the secondary part of $S$. The value $\mathit{val}_{S, \zeta}(t)$ of a term $t$ in the state $S$ under the valuation $\zeta$ is defined as usual in first-order logic. We distinguish among two classes of first-order formulae over metafinite structures (states). The first class, denoted $\mathrm{FO}$ is the class of first-order formulae with equality which is built up from equations between terms in ${\cal T}_{\Upsilon, {\cal X}}$ by using the standard connectives and first-order quantifiers. The second class, denoted $\mathrm{FO}^*$ is the class of first-order formulae with equality which is built up from equations between terms in ${\cal T}_{\Upsilon, {\cal X}_1}$, again by using the standard connectives and first-order quantifiers. The semantics of both classes of first-order formulae is defined in the standard way. The truth value of a first-order formula $\varphi$ in $S$ under the valuation $\zeta$ is denoted as $[\![\varphi]\!]_{S,\zeta}$.

Let $S$ be a meta finite state of signature $\Upsilon$. A \emph{valuation} or \emph{variable assignment} $\zeta$ is a function that assigns to every variable in ${\cal X}_1$ a value in the base set $B_1$ of the primary part of $S$ and to every variable in ${\cal X}_2$ a value in the base set $B_2$ of the secondary part of $S$. The value $\mathit{val}_{S, \zeta}(t)$ of a term $t \in {\cal T}_{\Upsilon, {\cal X}}$ in the state $S$ under the valuation $\zeta$ is defined as usual in first-order logic. The \emph{first-order logic of metafinite structures} (states) is defined as the first-order logic with equality which is built up from equations between terms in ${\cal T}_{\Upsilon, {\cal X}}$ by using the standard connectives and first-order quantifiers. Its semantics is defined in the standard way. The truth value of a first-order formula of meta finite structures $\varphi$ in $S$ under the valuation $\zeta$ is denoted as $[\![\varphi]\!]_{S,\zeta}$.

In our definition of ASM rule, we use the fact that function arguments can be read as tuples. Thus, if $f$ is an $n$-ary function and $t_1, \ldots, t_n$ are arguments for $f$, we write $f(t)$ where $t$ is a term which evaluates to the tuple $(t_1, \ldots, t_n)$, instead of $f(t_1, \ldots, t_n)$. This is not strictly necessary, but it greatly simplifies the presentation of the technical details in this paper.
Let $t$ and $s$ denote terms in ${\cal T}_p$, let $\mathtt{t}$ and $\mathtt{s}$ denote terms in ${\cal T}_a$ and let $\varphi$ denote a first-order formula of metafinite structures of vocabulary $\Upsilon$. The set of {\em ASM rules} over $\Upsilon$ is inductively defined as follows:

\begin{itemize}

\item {\em update rule $1$}: $f(t) := s$ (where $f \in \Upsilon_1$);

\item {\em update rule $2$}: $f(\mathtt{t}) := \mathtt{s}$ (where $f \in \Upsilon_2$);

\item {\em update rule $3$}: $f(t) := \mathtt{s}$ (where $f \in {\cal F}_b$); 

\item {\em conditional rule}: \textbf{if} $\varphi$ \textbf{then} $r$ \textbf{endif}

\item {\em forall rule}: \textbf{forall} $x$ \textbf{with} $\varphi$ \textbf{do} $r$ \textbf{enddo}

\item {\em bounded choice rule}: \textbf{choose} $x$ \textbf{with} $\varphi$ \textbf{do} $r$ \textbf{enddo}

\item {\em unbounded choice rule}: \textbf{choose} $\mathtt{x}$ \textbf{with} $\varphi$ \textbf{do} $r$ \textbf{enddo}

\item {\em parallel rule}: \textbf{par} $r_1$ $r_2$ \textbf{endpar} (execute the rules $r_1$ and $r_2$ in parallel);

\item {\em sequence rule}: \textbf{seq} $r_1$ $r_2$ \textbf{endseq} (first execute rule $r_1$ and then rule $r_2$).

%\item {\em let rule}: assign the location operator $\rho$ to the location $(f,t)$ and then aggregate all update values of $l$ yielded by the rule $r$;

%\hspace{3cm}\textbf{let} $(f,t) \!\rightharpoonup\! \rho$ \textbf{in}
%$r$ \textbf{endlet}

\end{itemize}

If $r$ is an ASM rule of signature $\Upsilon$ and $S$ is a state of $\Upsilon$, then we associate to them a set $\Delta(r,S,\zeta)$ of update sets which depends on the variable assignment $\zeta$.
Let $\zeta[x \mapsto a]$ denote the variable assignment which coincides with $\zeta$ except that it assigns the value $a$ to $x$. We formally define in Figure~\ref{fig:set} the sets of update sets yielded by the ASM rules. Items~$1$--$3$ in Figure~\ref{fig:set} correspond to the update rules~$1$--$3$, respectively. Each update rules yields a set which contains a single update set, which in turns contains a single update to a function of $S$. Depending on whether the function name $f$ belongs to $\Upsilon_1$, $\Upsilon_2$ or ${\cal F}_b$, the produced update corresponds to a function in the primary or secondary part of $S$ or to a bridge function, respectively. 
The choice rules introduce non-determinism. The bounded choice rule yields a finite set of update sets, since $x$ range over the (finite) primary part of $S$ (see item~6 in Figure~\ref{fig:set}). The unbounded choice rule yields a possibly infinite set of update sets (see item~7 in Figure~\ref{fig:set}). In this latter case, $\mathtt{x}$ range over the (possible infinite) secondary part of $S$ and it might happen that there are infinite valuations for $\mathtt{x}$ that satisfy the condition $\varphi$, each resulting in a different update set.   
All other rules only rearrange updates into different update sets. Update sets are explained in more detail in Section~\ref{sub:UpdateSets}. 

\begin{figure}[htb]
\fbox{\parbox{11.9cm}{
\begin{enumerate}

\item  $\Delta(f(t) := s,S,\zeta) = \{ \{  (f, (a),b) \} \}$ for $a = val_{S,\zeta}(t) \in B_1$ and $b = val_{S,\zeta}(s) \in B_1$\\

\item  $\Delta(f(\mathtt{t}) := \mathtt{s},S,\zeta) = \{ \{  (f, (a),b) \} \}$ for $a = val_{S,\zeta}(\mathtt{t}) \in B_2$ and $b = val_{S,\zeta}(\mathtt{s}) \in B_2$\\

\item  $\Delta(f(t) := \mathtt{s},S,\zeta) = \{ \{  (f, (a),b) \} \}$ for $a = val_{S,\zeta}(t) \in B_1$ and $b = val_{S,\zeta}(\mathtt{s}) \in B_2$\\

\item
$\Delta(\text{\textbf{if} }\varphi\text{ \textbf{then} } r\text{
\textbf{endif}},S,\zeta) = \begin{cases} \Delta(r,S,\zeta) &\text{if
}  [\!\![\varphi]\!\!]_{S,\zeta} = \mathrm{true}
\\ \{\emptyset\} &\text{otherwise}
\end{cases}$\\

\item
$\Delta(\text{\textbf{forall} } x \text{ \textbf{with} }\varphi \text{\textbf{ do} }r\text{ \textbf{enddo}},S,\zeta) \!=$ \\
\hspace*{3cm} $\{ \Delta_1 \cup \dots \cup \Delta_n \mid \Delta_i \in \Delta(r,S,\zeta[x \mapsto a_i]) \}$, \\
\hspace*{3cm} where $\{a_1 ,\dots, a_n \} = \{a_i \in B_1 \mid [\!\![\varphi]\!\!]_{S,\zeta[x \mapsto a_i]} = \mathit{true} \}$\\

\item
$\Delta(\text{\textbf{choose} } x \text{ \textbf{with} }\varphi \text{\textbf{ do} }r\text{ \textbf{enddo}},S,\zeta) =$\\
\hspace*{4cm} $\bigcup\limits_{a_i \in B_1}\{ \Delta(r,S,\zeta[x \mapsto a_i]) \mid [\!\![\varphi]\!\!]_{S,\zeta[x \mapsto a_i]} = \mathrm{true} \}$ \\

\item
$\Delta(\text{\textbf{choose} } \mathtt{x} \text{ \textbf{with} }\varphi \text{\textbf{ do} }r\text{ \textbf{enddo}},S,\zeta) =$\\
\hspace*{4cm} $\bigcup\limits_{a_i \in B_2}\{ \Delta(r,S,\zeta[x \mapsto a_i]) \mid [\!\![\varphi]\!\!]_{S,\zeta[x \mapsto a_i]} = \mathrm{true} \}$ \\

\item
$\Delta(\text{\textbf{par} }r_1 \text{ } r_2\text{ \textbf{endpar}},S,\zeta) =$\\
\hspace*{4.2cm} $\{ \Delta_1 \cup \Delta_2 \mid \Delta_1 \in \Delta(r_1,S,\zeta) \; \text{and} \; \Delta_2 \in \Delta(r_2,S,\zeta) \}$ \\

\item
$\Delta(\text{\textbf{seq} }r_1 \text{ } r_2 \text{
\textbf{endseq}},S,\zeta) =$ \\
\hspace*{1cm} $\{ \Delta_1 \oslash \Delta_2 \mid
\Delta_1 \in \Delta(r_1,S,\zeta) \;\text{is consistent and }\Delta_2 \in
\Delta(r_2,S+\Delta_1,\zeta)\} \cup$ \\
\hspace*{1cm} $\{ \Delta_1 \in \Delta(r_1,S,\zeta) \mid \Delta_1 \;\text{is inconsistent} \}$,\\
\hspace*{1cm} where $\Delta_1 \oslash \Delta_2 = \Delta_2 \cup \{
(\ell,a) \in \Delta_1 \mid  \ell \neq \ell^\prime \text{ for all }(\ell',a^\prime) \in \Delta_2\}$

%\item
%$\Delta(\text{\textbf{let} } (f,t) \!\rightharpoonup\!\rho \text{
%\textbf{in} }r\text{ \textbf{endlet}},S,\zeta) =$\\
%$\{\{ (\ell,a)
%\mid a = \rho( \{\!\!\{ a_i \mid (\ell,a_i) \in
%\ddot{\Delta}\}\!\!\} ) \} \cup \{ (\ell^\prime,b) \in \ddot{\Delta} \mid \ell^\prime \neq \ell \} \mid \ddot{\Delta}\in \ddot{\Delta}(r,S,\zeta%)\}$
%\hspace{6.3cm}where $\ell = (f, val_{S,\zeta}(t))$

\end{enumerate}}}\caption{Sets of update sets of non-deterministic parallel ASMs}\label{fig:set}
\end{figure}

\begin{remark}
For every state $S$, ASM rule $r$ and variable assignment $\zeta$, we have that every $\Delta \in \Delta(r,S,\zeta)$ is a finite set of updates. This is a straightforward consequence of the fact that the variable $x$ in the definition of the $\textbf{forall}$ rule ranges over the (finite) primary part of $S$, and it is also the case in the ASM thesis for parallel algorithms of Blass and Gurevich~\cite{blass:tocl2003,GurevichParallelCorrection08} where it is implicitly assumed that the $\textbf{forall}$ rule in the parallel ASMs range over finite hereditary multisets. See our work in~\cite{FerrarottiSTW15} for a detailed explanation. 
Regarding the set $\Delta(r,S,\zeta)$ of update sets, we note that it might be infinite since the unbounded choice rule can potentially produce infinitely many update sets. In fact, this is the case if we consider the first unbounded choice rule in Example~\ref{ex-diff-words}.  
\end{remark}

Formally, a {\em non-deterministic parallel} ASM $M$ over a signature $\Upsilon$ of metafinite states consists of: (a) a set $\mathcal{S}$ of metafinite states over $\Upsilon$, (b) non-empty subsets $\mathcal{S}_I \subseteq \mathcal{S}$ of \emph{initial states} and $\mathcal{S}_F \subseteq \mathcal{S}$ of \emph{final states}, and (c) a \emph{closed} ASM rule $r$ over $\Upsilon$, i.e., a rule $r$ in which all free variables in the first-order formulae of the rule are bounded by $\textbf{forall}$ or $\textbf{choose}$ constructs.

Every non-deterministic parallel ASM $M$ defines a corresponding \emph{successor relation} $\delta$ over $\mathcal{S}$ which is determined by the main rule $r$ of $M$. A pair of states $(S_1, S_2)$ belongs to $\delta$ iff there is a consistent update set $\Delta \in \Delta(r,S)$ (the valuation $\zeta$ is omitted from $\Delta(r,S,\zeta)$ since $r$ is closed) such that $S_2$ is the unique state resulting from updating $S_1$ with $\Delta$.
A {\em run} of an ASM $M$ is a finite sequence $S_0 ,\dots, S_n$ of states with
$S_0 \in \mathcal{S}_I$, $S_n \in \mathcal{S}_F$, $S_i \notin
\mathcal{S}_F$ for $0 < i < n$, and $(S_i,S_{i+1}) \in \delta$ for
all $i=0,\dots,n-1$.

The following example, adapted from~\cite{Huggins02}, illustrates a parallel ASMs with bounded non-determinism.  

\begin{example}\label{example1}
We consider metafinite states with: (a) a primary part formed by a connected weighted graph $G=(V,E)$, (b) a secondary part formed by the set of natural numbers $\mathbb{N}$, and (c) a bridge function $\mathit{weight}$ from the set of edges in $E$ to $\mathbb{N}$. Apart from the static (Boolean) function symbols $V$ and $E$, the vocabulary of the primary part of the states also includes dynamic function symbols $\mathit{label}$ and $T$, and  static function symbols $\mathit{first}$ and $\mathit{second}$, the last two for extracting the first and second element of an ordered pair, respectively. Since $G$ is an undirected graph, we have that $(x,y) \in E$ iff $(y,x) \in E$.

The non-deterministic parallel ASM in this example, which we denote as $M$, formally expresses Kruskal's algorithm~\cite{Kruskal56} for computing the \emph{minimum spanning tree} in a connected, weighted graph. Recall that a spanning tree $T$ of a graph $G$ is a tree such that every pair of nodes in $G$ are connected via edges in $T$. We say that $T$ is minimum if the sum of the weights of all its edges is the least among all spanning trees of $G$. We assume that in every initial state of $M$, $\mathit{label}(x) = x$ for every $x \in V$ and that $T((x,y)) = \mathit{false}$ for every $(x,y) \in E$.

The condition in the first \textbf{choose} rule is simply ensuring that the chosen edge $x$ is eligible, i.e., that the nodes $\mathit{first}(x)$ and $\mathit{second}(x)$ that make up the endpoints of the edge $x$ have different labels, and that $x$ has minimal weight among the set of eligible edges. The following two update rules simply add the edge $x$ to the tree $T$. The second \textbf{choose} rule reflects the fact that from the point of view of the correctness of the algorithm, it does not matter which endpoint $y$ of the edge $x$ we choose at this stage. Finally, the \textbf{forall} rule simply relabels (as expected) every node with the same label than the endpoint $y$ of $x$ (including the node $y$ itself) with the label of the opposite endpoint of $x$. \\[-0.8cm]

\begin{alltt}
\small
choose \(x\) with \(E(x)\land\mathit{label}(\mathit{first}(x))\neq\mathit{label}(\mathit{second}(x))\land\)
         \(\forall\)\(y (E(y)\land\mathit{label}(\mathit{first}(y))\neq\mathit{label}(\mathit{second}(y))\rightarrow\mathit{weight}(y)\ge\mathit{weight}(x))\) do
  \(\mathit{T}(x):=\mathit{true}\)
  \(\mathit{T}((\mathit{second}(x),\mathit{first}(x))):=\mathit{true}\)
  choose \(y\) with \(y=\mathit{first}(x)\lor{y}=\mathit{second}(x)\) do
      forall \(z\) with \(\mathit{label}(z)=\mathit{label}(y)\) do
          if \(\mathit{label}(y)=\mathit{label}(\mathit{first}(x))\) then \(\mathit{label}(z):=\mathit{label}(\mathit{second}(x))\) endif
          if \(\mathit{label}(y)=\mathit{label}(\mathit{second}(x))\) then \(\mathit{label}(z):=\mathit{label}(\mathit{first}(x))\) endif
      enddo
  enddo
enddo
\end{alltt}
\end{example}

%%%%%%%%%%%%%%%%%%%%%%%%%%%%%%%%%%%%%%%%%%%%%%%%%%%%%%%%%%%%%%%
\section{A Logic for Non-Deterministic Parallel ASMs}
\label{sec:TheLogic}

The logic for non-deterministic parallel ASMs (denoted ${\cal L}$) is a dynamic first-order logic extended with membership predicates over finite sets, an update set predicate and a multi-modal operator. ${\cal L}$ is defined over many sorted first-order structures which have:
\begin{itemize}
\item a \emph{finite individual sort} with variables $x_1,x_2,...$ which range over a finite domain $D_1$,
\item an \emph{individual sort} with variables $\mathtt{x}_1, \mathtt{x}_2, \ldots$, which range over a (possibly infinite) domain $D_2$, and
\item a \emph{predicate sort} with variables $x^1_1, x^1_2, \ldots,$ which range over the domain $P_1$ formed by all finite subsets (relations) on
$\mathcal{F}_{dyn} \times (D_1 \cup D_2) \times (D_1 \cup D_2)$.
\item a \emph{predicate sort} with variables $x^2_1, x^2_2, \ldots,$ which range over the domain $P_2$ formed by all finite subsets (relations) on $\mathcal{F}_{dyn} \times (D_1 \cup D_2) \times (D_1 \cup D_2) \times D_1$.
\end{itemize}
A signature $\Sigma$ of the logic ${\cal L}$ comprises a finite set $F_1$ of names for functions on $D_1$, a finite set $F_2$ of names for functions on $D_2$, and a finite set $F_b$ of names for functions which take arguments from $D_1$ and return values on $D_2$.

We define terms of ${\cal L}$  by induction. Variables $x_1,x_2,...$ and $\mathtt{x}_1, \mathtt{x}_2, \ldots$ are terms of the first and second individual sort, respectively. Variables $x^1_1, x^1_2, \ldots$ and $x^2_1, x^2_2, \ldots$ are terms of the first and second predicate sort, respectively. If $f$ is an $n$-ary function name in $F_1$ and $t_1, \ldots, t_n$ are terms of the first individual sort, then $f(t_1, \ldots, t_n)$ is a term of the first individual sort. If $f$ is an $n$-ary function name in $F_2$ and $t_1, \ldots, t_n$ are terms of the second individual sort, then $f(t_1, \ldots, t_n)$ is a term of the second individual sort. If $f$ is an $n$-ary function name in $F_b$ and $t_1, \ldots, t_n$ are terms of the first individual sort, then $f(t_1, \ldots, t_n)$ is a term of the second individual sort.

The formulae of ${\cal L}$ are those generated by the following grammar:
\[\begin{aligned}
\varphi, \psi \, ::= & \, s=t \mid s_a = t_a \mid  \neg\varphi \mid \varphi \wedge \psi \mid \forall x (\varphi) \mid \forall \texttt{x} (\varphi) \mid \forall x^1 (\varphi) \mid \forall x^2 (\varphi) \mid \\
&\in^1\!\!(x^1\!,f,t_0,s_0) \mid \in^2\!\!(x^2\!,f,t_0,s_0,s) \mid \mathrm{upd}(r, x^1) \mid [x^1]\varphi
\end{aligned}\]
where $s$ and $t$ denote terms of the first individual sort, $s_a$ and $t_a$ denote terms of the second individual sort, $f$ is a dynamic function symbol, $r$ is an ASM rule and, $t_0$ and $s_0$ denote terms of either the first or the second individual sort.

The interpretation of terms and the semantics of the first-order formulae is defined in the standard way. This includes equality which is used under a fixed interpretation and only between terms of a same individual sort.

The update set predicate $\mathrm{upd}(r, x^1)$ states that the \emph{finite} update set represented by $x^1$ is generated by the rule $r$. Let $S$ be a state of some signature $\Sigma$ of the logic ${\cal L}$. Let $\zeta$ be a variable assignment over $S$ which maps each variable of the first and second individual sort to a value in $D_1$ and $D_2$, respectively, and maps each variable of the first and second predicate sort to a value in $P_1$ and $P_2$, respectively. The truth value of $\mathrm{upd}(r, x^1)$ is defined by $[\![\mathrm{upd}(r, x^1)]\!]_{S,\zeta} = \mathit{true}$ iff $\mathit{val}_{S,\zeta}(x^1) \in \Delta(r,S,\zeta)$.

The set membership predicate $\in^1\!\!(x^1\!,f,t_0,s_0)$ indicates that $(f,t_0,s_0)$ is an update in the update set represented by $x^1$ while the auxiliary set membership predicate $\in^2\!\!(x^2\!,f,t_0,s_0,s)$ is used to keep track of which parallel branch produced each update in $x^2$. Their truth values are formally defined as follows:\\
$[\![\in^1\!\!(x^1\!,f,t_0,s_0)]\!]_{S,\zeta}=\mathit{true} \; \text{iff} \; (f, \mathit{val}_{S,\zeta}(t_0), \mathit{val}_{S,\zeta}(s_0)) \in \mathit{val}_{S,\zeta}(x^1)$\\
$[\![\in^2\!\!(x^2\!,f,t_0,s_0,s)]\!]_{S,\zeta}=\mathit{true} \; \text{iff} (f, \mathit{val}_{S,\zeta}(t_0), \mathit{val}_{S,\zeta}(s_0), \mathit{val}_{S,\zeta}(s)) \in \mathit{val}_{S,\zeta}(x^2)$

Finally, we use $[x^1]\varphi$ to express the evaluation of $\varphi$ over the successor state obtained by applying the updates in $x^1$ to the current state. Its truth value is defined by: $[\![[x^1]\varphi]\!]_{S,\zeta}=\mathit{true}$ iff $\Delta = \zeta(x^1)$ is inconsistent or $[\![\varphi]\!]_{S+\Delta,\zeta}=\mathit{true}$ for $\zeta(x^1) = \Delta \in \Delta(r,S,\zeta)$. That is, when $\Delta = \zeta(x^1)$ is inconsistent, successor states for the current state $S$ do not exist and thus $S+\Delta$ is undefined. In this case, $[x^1]\varphi$ is interpreted as $\mathit{true}$. With the use of the modal operator [\hspace{0.1cm}] for an update set $\Delta=\zeta(x^1)$ (i.e., $[x^1]$), ${\cal L}$ is empowered to be a multi-modal logic. 
%The formulae of ${\cal L}$ are interpreted in states represented as a Kripke frame. A \emph{Kripke frame} is a pair $(U,R)$ consisting of a universe $U$ that is a
%non-empty set of states, and a binary accessibility relation $R$ on $U$ such that $(S,S^{\prime})\in R$ for $S,S^{\prime}\in U$.

We say that a formula $\varphi$  of ${\cal L}$ is \emph{static} if all the function symbols which appear in $\varphi$ are static and say that it is \emph{pure} if it is generated by the following grammar: $\varphi, \psi \, ::= \, s=t \mid s_a = t_a \mid  \neg\varphi \mid \varphi \wedge \psi \mid \forall x (\varphi) \mid \forall \texttt{x} (\varphi)$.

Since metafinite states are just a special kind of two sorted first-order structures in which one of the sorts is finite, we can identify every metafinite state $S$ of ${\cal L}$ with a corresponding many sorted first-order structure $S'$ of the class used in definition of ${\cal L}$. This can be done by taking the domains $D_1$ and $D_2$ of the individual sorts of $S'$ to be the base sets $B_1$ and $B_2$ of $S$, respectively, the sets $F_1$, $F_2$ and $F_b$ of function names of the signature $\Sigma$ of $S'$ to be the sets $\Upsilon_{1}$, $\Upsilon_2$ and ${\cal F}_b$ of the signature $\Upsilon$ of $S$, respectively, and the interpretation in $S'$ of the function names in $\Sigma$ to coincide with the interpretation in $S$ of the corresponding function symbols in $\Upsilon$. Following this transformation we have that for every state $S$, every corresponding pair of many sorted first-order structure $S'$ and $S''$ are isomorphic by an isomorphism which is the identity among elements of the individual sorts. Thus, we can talk of \emph{the} many sorted structure $S$ corresponding to a state $S$ and, when it is clear from the context, we can even talk of the state $S$ meaning the many sorted structure $S$.

In what follows, we use the somehow clearer and more usual syntax of second-order logic to denote the set membership predicates and the quantification over the predicate sorts. Thus we use upper case letters $X, Y, \ldots$ and ${\cal X}, {\cal Y}, \ldots$ to denote variables $x^1_1, x^1_2, \ldots$ and $x^2_1, x^2_2, \ldots$ of the first and second predicate sorts, respectively, and we write $\forall X(\varphi)$, $\forall {\cal X}(\varphi)$, $[X]\varphi$, $X(f,t_0,s_0)$, ${\cal X}(f,t_0,s_0,s)$ and $\mathrm{upd}(r, X)$ instead of $\forall x^1\,(\varphi)$, $\forall x^2\,(\varphi)$, $[x^1]\varphi$, $\in^1\!\!(x^1\!,f,t_0,s_0)$, $\in^1\!\!(x^1\!,f,t_0,s_0,s)$ and $\mathrm{upd}(r, x^1)$, respectively.
Furthermore, in our formulae we use disjunction $\vee$, implication $\rightarrow$, double implication $\leftrightarrow$ and existential quantification $\exists$. All of them are defined as abbreviations in the usual way.

\begin{example}
${\cal L}$ can express properties of the ASM in Example~\ref{example1} such as:
\begin{itemize}
\item If $r$ yields in the current state $S$ an update set $\Delta$ with an update $(T, x, \textit{true})$, then in the successor state $S+\Delta$ the vertices of $x$ have a same label.\\[0.2cm]
\hspace*{-0.4cm} $\forall X (\mathrm{upd}(r,X) \rightarrow \forall x (X(T,x,\mathit{true}) \rightarrow [X](\mathit{label}(\mathit{first}(x)) = \mathit{label}(\mathit{second}(x)))))$\\
\item Each update set yielded by $r$ updates $T$ in no more than one location. \\[0.2cm]
\hspace*{-0.4cm} $\forall X (\mathrm{upd}(r,X) \rightarrow \neg(\exists x y (X(T,x,\mathit{true}) \wedge X(T,y,\mathit{true}) \wedge x \neq y)))$\\
\item If an edge $x$ meets in a state $S$ the criteria of the first \textbf{choose} rule in $r$, then there is an update set $\Delta \in \Delta(r,S)$ such that $T(x) = \textit{true}$ holds in $S+\Delta$.\\[0.2cm]   
\hspace*{-0.4cm} $\forall x (E(x)\land\mathit{label}(\mathit{first}(x))\neq\mathit{label}(\mathit{second}(x))\land$\\[0.1cm]
\hspace*{0.2cm} $\forall y (E(y)\land\mathit{label}(\mathit{first}(y))\neq\mathit{label}(\mathit{second}(y))\rightarrow\mathit{weight}(y)\ge\mathit{weight}(x))$\\[0.1cm]
\hspace*{0.2cm} $\rightarrow \exists X (\mathit{upd}(r, X) \wedge [X](T(x) = \textit{true}))) $
\end{itemize}
\end{example}

\section{A Proof System} \label{sec:proof_system}

In this section we develop a proof system for the logic ${\cal L}$ for non-deterministic parallel ASMs.

\begin{definition}\label{def-implied-formula}
We say that a state $S$ is a \emph{model} of a formula $\varphi$ (denoted as $S \models \varphi$) iff $[\![\varphi]\!]_{S,\zeta}= \textit{true}$ holds for every variable assignment $\zeta$.  If $\Psi$ is a set of formulae, we say that $S$ \emph{models} $\Psi$ (denoted as $S \models \Psi$) iff $S \models \varphi$ for each $\varphi \in \Psi$.
A formula $\varphi$ is said to be a \emph{logical consequence} of a set $\Psi$ of formulae (denoted as $\Psi\models\varphi$) if for every state $S$, if $S \models \Psi$, then $S \models \varphi$.
A formula $\varphi$ is said to be \emph{valid} (denoted as $\models \varphi$) if $[\![\varphi]\!]_{S,\zeta}=true$ in every state $S$ for every variable assignment $\zeta$.
A formula $\varphi$ is said to be \emph{derivable} from a set $\Psi$ of formulae (denoted as $\Psi\vdash_{\mathfrak{R}}\varphi$) if there is a deduction from formulae in $\Psi$
to $\varphi$ by using a set $\mathfrak{R}$ of axioms and inference rules.
\end{definition}

We will define such a set $\mathfrak{R}$ of axioms and rules in Subsection \ref{sub:AxiomsRules}. Then we simply write $\vdash$ instead of $\vdash_{\mathfrak{R}}$. We also define equivalence between two ASM rules. Two equivalent
rules $r_1$ and $r_2$ are either both defined or both undefined.

\begin{definition} \label{def-equivalent-rules}Let $r_1$ and $r_2$ be two ASM rules. Then
$r_1$ and $r_2$ are \emph{equivalent} (denoted as $r_1\equiv r_2$)
if for every state $S$ it holds that $S \models \forall X (\mathrm{upd}(r_1,X) \leftrightarrow \mathrm{upd}(r_2,X))$.
\end{definition}

%The substitution of a term $t$ for a variable $x$ in a formula
%$\varphi$ (denoted as $\varphi[t/x]$) is defined by the rule of
%substitution. That is, $\varphi[t/x]$ is the result of replacing all
%free instances of $x$ by $t$ in $\varphi$ provided that no free
%variable of $t$ becomes bound after substitution.

\subsection{Consistency}\label{sub:Consistency}

In \cite{RobertLogicASM} Nanchen and St\"ark use a predicate $\mathrm{Con}(r)$ as an
abbreviation for the statement that the rule $r$ is consistent. As every rule $r$ in their work is
deterministic, there is no ambiguity with the reference to the
update set associated with $r$, i.e., each deterministic rule $r$ generates exactly one (possibly empty) update set. Thus a deterministic rule $r$ is consistent iff the update set generated by $r$ is consistent.
However, in our logic ${\cal L}$, the presence of non-determinism
makes the situation less straightforward.

Let $r$ be an ASM rule and $\Delta$ be an update set. Then the
consistency of an update set $\Delta$, denoted by the formula $\mathrm{conUSet}(X)$ (where $X$ represents $\Delta$), can be expressed as:
\begin{equation}\label{con}
\mathrm{conUSet}(X)\equiv\bigwedge\limits_{f\in
\mathcal{F}_{dyn}}\forall x y z ((X(f,x,y) \wedge X(f, x, z)) \rightarrow y=z)
\end{equation}
Then $\mathrm{con}(r,X)$ is an abbreviation of the following formula
which expresses that an update set $\Delta$ (represented by the variable $X$) generated by the rule $r$
is consistent.
\begin{equation}\label{conr}
\text{con}(r,X)\equiv\mathrm{upd}(r,X)\wedge\mathrm{conUSet}(X)
\end{equation}

As the rule $r$ may be non-deterministic, it is possible that $r$
yields several update sets. Thus, we develop the consistency of ASM rules in
two versions:

\begin{itemize}

\item A rule $r$ is \emph{weakly consistent} (denoted as
$\mathrm{wcon}(r)$) if at least one update set generated
by $r$ is consistent. This can be expressed as follows:
\begin{equation}\label{wcon}
  \text{wcon}(r)\equiv \exists X (\mathrm{con}(r,X))
\end{equation}

\item A rule $r$ is \emph{strongly consistent} (denoted as
scon$(r)$) if every update set generated by $r$
is consistent. This can be expressed as follows:
\begin{equation}\label{scon}
  \text{scon}(r)\equiv \forall X (\text{upd}(r,X)\Rightarrow\text{con}(r,X))
\end{equation}

\end{itemize}

In the case that a rule $r$ is deterministic, the weak notion of
consistency coincides with the strong notion of consistency, i.e.,
$\mathrm{wcon}(r) \leftrightarrow \mathrm{scon}(r)$.

\subsection{Update Sets}\label{sub:UpdateSets}

We present the axioms for the predicate $\mathit{upd}(r,X)$ in Figure~\ref{Fig-AxiomsUpdateSets}.
To simplify the presentation, we give the formulae only for the case in which all the function symbols in ${\cal F}_{\mathit{dyn}}$ correspond to functions on the primary part (finite individual sort) of the state. To deal with dynamic function symbols corresponding to function of the secondary part and to bridge functions, we only need to slightly change the formulae by replacing some of the first-order variables in ${\cal X}_{1}$ by first-order variables in ${\cal X}_{2}$. For instance, if $f$ is a bridge function symbol, we should write $\forall x \mathtt{y} (X(f, x, \mathtt{y}) \rightarrow x = t \wedge \mathtt{y} = \mathtt{s})$ instead of $\forall x y (X(f, x, y) \rightarrow x = t \wedge y = s)$.

%, where inv$(X,f,\bar{x})$ and
%inv$(\ddot{X},f,\bar{x})$ asserting that update sets
%$\Delta$ represented by $X$ and update multisets $\ddot{\Delta}$ represented by $\ddot{X}$ do not have any update to the
%location of the function symbol $f$ at the argument $\bar{x}$,
%respectively. This can be formalised as follows:

%\begin{center}

% inv$(X,f,\bar{x})\equiv$ $\forall y.\neg X(f,\bar{x},y)$

%inv$(\ddot{X},f,\bar{x})\equiv$ $\forall y,z.\neg\ddot{X}(f,\bar{x},y,z)$.

%\end{center}

\begin{figure}[h!]
\begin{tabular*}{1\textwidth}{@{\extracolsep{\fill}}l}\hline\hline
  \\
\textbf{U1}. $\mathrm{upd}(f(t) := s, X) \leftrightarrow X(f,t,s) \wedge \forall x y (X(f, x, y) \rightarrow x = t \wedge y = s) \wedge$ \\[0.2cm]
\hspace{4cm}$\bigwedge\limits_{f \neq f^{\prime} \in\mathcal{F}_{dyn},}\forall x y (\neg X(f^{\prime},x,y))$  \\
\\
\textbf{U2}. $\mathrm{upd}(\textbf{if} \, \varphi \, \textbf{then}\, r \,\textbf{endif}, X) \leftrightarrow (\varphi \wedge \mathrm{upd}(r,X)) \vee (\neg \varphi \wedge \bigwedge\limits_{f\in \mathcal{F}_{dyn}}\forall x y (\neg X(f,x,y))$ \\
\\
\textbf{U3}. $\mathrm{upd}(\textbf{forall} \, x \, \textbf{with} \, \varphi \, \textbf{do} \, r \, \textbf{enddo},X) \leftrightarrow$\\[0.2cm]
\hspace{0.8cm}$\exists {\cal X} \big(\forall x \big( (\varphi \rightarrow \exists Y (\mathrm{upd}(r,Y) \wedge \bigwedge\limits_{f\in \mathcal{F}_{dyn}} \forall y_1 y_2 (Y(f,y_1,y_2) \leftrightarrow {\cal X}(f,y_1,y_2,x)))) \wedge$\\[0.2cm]
\hspace{2.15cm}$( \neg \varphi \rightarrow \bigwedge\limits_{f\in \mathcal{F}_{dyn}} \forall y_1 y_2 (\neg {\cal X}(f,y_1,y_2,x))) \big) \wedge$\\[0.2cm]
\hspace{1.5cm}$\bigwedge\limits_{f \in \mathcal{F}_{dyn}} \forall x_1 x_2 (X(f, x_1, x_2) \leftrightarrow \exists x_3 ({\cal X}(f, x_1, x_2, x_3)))\big)$\\
\\
\textbf{U4}. $\mathrm{upd}(\textbf{par} \, r_1 \; r_2 \, \textbf{endpar}, X) \leftrightarrow \exists Y_1 Y_2 (\mathrm{upd}(r_1,Y_1) \wedge \mathrm{upd}(r_2,Y_2) \wedge$\\[0.2cm]
   \hspace{5cm}$\bigwedge\limits_{f\in\mathcal{F}_{dyn}} \forall x y (X(f,x,y) \leftrightarrow (Y_1(f, x,y) \vee Y_2(f,x,y)))$\\
 \\
\textbf{U5}. $\mathrm{upd}(\textbf{choose}\, x \, \textbf{with} \, \varphi \, \textbf{do} \, r \, \textbf{enddo}, X) \leftrightarrow \exists x (\varphi \wedge \mathrm{upd}(r,X))$ \\
 \\
\textbf{U6}. $\mathrm{upd}(\textbf{choose}\, \mathtt{x} \, \textbf{with} \, \varphi \, \textbf{do} \, r \, \textbf{enddo}, X) \leftrightarrow \exists \mathtt{x} (\varphi \wedge \mathrm{upd}(r,X))$ \\
 \\
\textbf{U7}. $\mathrm{upd}(\textbf{seq} \, r_1 \; r_2 \, \textbf{endseq}, X) \leftrightarrow \big(\text{upd}(r_1,X) \wedge \neg\text{con}(X)\big) \vee$\\[0.2cm]
\hspace{0.8cm}$\big( \exists Y_1 Y_2 (\mathrm{upd}(r_1,Y_1) \wedge \text{con}(Y_1) \wedge [Y_1]\mathrm{upd}(r_2,Y_2) \wedge$\\[0.2cm]
\hspace{1.6cm} $\bigwedge\limits_{f\in\mathcal{F}_{dyn}} \forall x y (X(f,x,y) \leftrightarrow ((Y_1(f,x,y) \wedge \forall z (\neg Y_2(f,x,z))) \vee Y_2(f,x,y))))\big)$ \\
% \\
%\textbf{U7}. $\mathrm{upd}(\textbf{let} \, (f,t)\!\rightharpoonup\!\rho \, \textbf{in} \, r \,\textbf{endlet},X) \leftrightarrow$\\[0.2cm]
%\hspace{0.7cm} $\exists \ddot{X} \big(\mathrm{upm}(r,\ddot{X}) \wedge \forall x y (X(f,x,y) \leftrightarrow ((t = x \wedge y=\rho_{y^{\prime}}(y^{\prime}|\exists \mathtt{z} (\ddot{X}(f,x,y^{\pr%ime},\mathtt{z})))) \vee$\\[0.2cm]
%\hspace{6.2cm} $(t \neq x \wedge \exists \mathtt{z} (\ddot{X}(f,x,y,\mathtt{z}))))) \wedge$\\[0.2cm]
%\hspace{1.2cm} $\bigwedge\limits_{f \neq f' \in \mathcal{F}_{dyn}}\forall x y (X(f',x,y) \leftrightarrow \exists \mathtt{z} (\ddot{X}(f',x,y,\mathtt{z})))\big)$
\\
\hline\hline
\end{tabular*}

\caption{\label{Fig-AxiomsUpdateSets} Axioms for predicate upd(r,$X$)}
\end{figure}

In the following we explain Axioms~\textbf{U1}-\textbf{U7} in turn. We assume a state $S$ of some signature $\Upsilon$ and base set $B = B_1 \cup B_2$, where $B_1$ is the base set of the \emph{finite} primary part of $S$. We also assume a variable assignment $\zeta$.

As in our case an ASM rule may be non-deterministic, a straightforward
extension from the formalisation of the $\textbf{forall}$ and $\textbf{par}$ rules
used in the logic for ASMs in \cite{RobertLogicASM} would not work for
Axioms \textbf{U3} and \textbf{U4}. The axioms correspond to the definition of update sets in Figure~\ref{fig:set}.

\begin{itemize}

\item Axiom \textbf{U1} says that $X$ is an update yielded by the assignment
rule $f(t) := s$ iff it contains exactly one update which is $(f,t,s)$.

\item Axiom \textbf{U2} asserts that, if the formula $\varphi$ evaluates to $\mathit{true}$, then
$X$ is an update set yielded by the conditional rule \textbf{if} $\varphi$
\textbf{then} $r$ \textbf{endif} iff $X$ is an update set yielded by the rule
$r$. Otherwise, the conditional rule yields only an empty update
set.

\item Axiom \textbf{U3} states that $X$ is an update set yielded by the rule \textbf{forall} $x$ \textbf{with} $\varphi$ \textbf{do} $r$ \textbf{enddo} iff
$X$ coincides with $\Delta_{a_1} \cup \cdots \cup \Delta_{a_n}$, where $\{a_1, \ldots, a_n\} = \{ a_i \in B_{1} \mid val_{S,\zeta[x \mapsto a_i]}(\varphi) = \mathit{true}\}$ and $\Delta_{a_i}$ (for $1 \leq i \leq n$) is an update set yielded by the rule $r$ under the variable assignment $\zeta[x \mapsto a_i]$. Note that the update sets $\Delta_{a_1}, \ldots, \Delta_{a_n}$ are encoded into ${\cal X}$.

\item Axiom \textbf{U4} states that $X$ is an update set yielded by the parallel rule \textbf{par} $r_1\hspace{0.2cm} r_2$ \textbf{endpar} iff it corresponds to the
union of an update set yielded by $r_1$ and an update set yielded by $r_2$.

\item Axioms~\textbf{U5} asserts that $X$ is an update set yielded by the rule \textbf{choose} $x$ \textbf{with} $\varphi$
\textbf{do} $r$ \textbf{enddo} iff it is an update set yielded by the rule $r$ under a variable assignment $\zeta[x \mapsto a]$ which satisfies $\varphi$.

\item Axiom \textbf{U6} is similar to Axiom~\textbf{U5}, but for the case of the \textbf{choose} $\mathtt{x}$ \textbf{with} $\varphi$
\textbf{do} $r$ \textbf{enddo} rule.

\item Axiom \textbf{U7} asserts that $X$ is an update set yielded by a sequence rule \textbf{seq} $r_1\hspace{0.2cm} r_2$ \textbf{endseq}
iff it corresponds to either an inconsistent update set yielded by rule $r_1$, or to an update set formed by the updates in an update set $Y_2$ yielded by rule $r_2$ in a successor state $S+Y_1$, where $Y_1$ encodes a consistent set of updates produced by rule $r_1$, plus the updates in $Y_1$ that correspond to locations other than the locations updated by $Y_2$.

%\item Axiom \textbf{U7} asserts that an update multiset is
%collapsed into an update set by aggregating the values of the
%location $(f,t)$ that has been assigned a location operator $\rho$, and
%ignoring the multiplicity of updates if their locations are not $t$.
\end{itemize}

The following lemma is an easy consequence of the axioms in Figure~\ref{Fig-AxiomsUpdateSets}.

\begin{lemma}\label{lem-upd}

Every formula in the logic ${\cal L}$ can be replaced by an equivalent formula not containing any subformulae of the form $\mathrm{upd}(r,X)$.

\end{lemma}

\begin{remark}
The inclusion of the parameter $X$ in the
predicate upd$(r,X)$ is important because a rule $r$ in a non-deterministic parallel ASM rule may be associated with multiple update sets,
and thus we need a way to specify which update set yielded by rule $r$ is meant.
\end{remark}

\subsection{Axioms and Inference Rules}\label{sub:AxiomsRules}

Now we can present a set of axioms and inference rules which constitute a proof system for the logic ${\cal L}$. To avoid unnecessary repetitions of almost identical axioms and rules, we describe them only considering variables of the first individual sort, but the exact same axioms and inference rules are implicitly assumed for the case of variables of the second individual sort as well as for variables of the predicate sorts. In the definition of the set of axioms and rules, we sometimes use $\varphi[t/x]$ to denote the substitution of a term $t$ for a variable $x$ in a formula $\varphi$. That is, $\varphi[t/x]$ is the result of replacing all free instances of $x$ by $t$ in $\varphi$ provided that no free variable of $t$ becomes bound after substitution.

Formally, the set $\mathfrak{R}$ of axioms and inference rules is formed by:
\begin{itemize}
%\item The axioms \textbf{D1}-\textbf{D7} in Fig. \ref{Fig-AxiomsUndefinedness} assert the properties of def$(r)$.

%\smallskip

\item The axioms \textbf{U1}-\textbf{U7} in Fig. \ref{Fig-AxiomsUpdateSets} which assert the properties of upd$(r, X)$.

\smallskip

%\item The axioms \textbf{UM1}-\textbf{UM7} in Fig. \ref{Fig-AxiomsUpdateMultisets} which assert the properties of upm$(r, \ddot{X})$.

%\smallskip

\item Axiom \textbf{M1} and Rules \textbf{M2-M3} from the axiom system K of modal logic, which is the weakest
normal modal logic system \cite{hughes:modallogic1996}. Axiom~\textbf{M1} is called \emph{Distribution Axiom} of K,
Rule~\textbf{M2} is called \emph{Necessitation Rule} of K and Rule~\textbf{M3} is the inference rule called \emph{Modus Ponens} in the
classical logic. By using these axiom and rules together, we are able to derive all modal properties that are valid in Kripke frames.

\begin{description}

  \item[\textbf{M1}] $[X](\varphi\rightarrow\psi)\rightarrow ([X]\varphi\rightarrow[X]\psi)$\smallskip

  \item[\textbf{M2}] $\varphi\vdash [X]\varphi$ \hspace*{3.5cm} \textbf{M3} $\varphi,\varphi\rightarrow\psi\vdash\psi$\smallskip
\end{description}

\item Axiom \textbf{M4} asserts that, if an update set $\Delta$ is not consistent, then
there is no successor state obtained after applying $\Delta$ over
the current state and thus $[X]\varphi$ (for $X$ interpreted by $\Delta$) is interpreted as true
for any formula $\varphi$. As applying a consistent update set $\Delta$
over the current state is deterministic, Axiom \textbf{M5} describes
the deterministic accessibility relation in terms of $[X]$.

\begin{description}

  \item[\textbf{M4}] $\neg \mathrm{conUSet}(X)\rightarrow[X]\varphi $  \hspace*{1.5cm} \textbf{M5} $\neg[X]\varphi\rightarrow [X]\neg\varphi$ \smallskip

\end{description}

\item Axiom \textbf{M6} is called \emph{Barcan Axiom}. It originates from the fact that all
states in a run of a non-deterministic parallel ASM have the same base set, and thus the
quantifiers in all states always range over the same set of
elements.

\begin{description}

\item[\textbf{M6}] $\forall x ([X]\varphi)\rightarrow [X]\forall
x(\varphi)$ \smallskip

\end{description}

\item Axioms \textbf{M7} and \textbf{M8} assert that the interpretation of static or pure formulae is the same in all states of non-deterministic parallel ASMs, since they are not affected by the execution of any ASM rule $r$.
%Note that, depending on the logic that is parameterised into the logic of meta-finite states, static and pure formulae might not be first-order formulae.

\begin{description}
  \item[\textbf{M7}] con$(r,X)\wedge\varphi\rightarrow [X]\varphi$ for
  static or pure $\varphi$\smallskip

  \item[\textbf{M8}]  con$(r,X)\wedge[X]\varphi\rightarrow
  \varphi$ for static or pure $\varphi$\smallskip

\end{description}

\item Axiom \textbf{A1} asserts that, if a consistent update set $\Delta$ (represented by $X$) does
not contain any update to the location $(f,x)$, then the
content of $(f,x)$ in a successor state obtained after
applying $\Delta$ is the same as its content in the current state.
Axiom \textbf{A2} asserts that, if a consistent update set $\Delta$
does contain an update which changes the content of the location
$(f,x)$ to $y$, then the content of $(f,x)$ in
the successor state obtained after applying $\Delta$ is $y$.
%Axiom \textbf{A3} says that, if a DB-ASM
%rule $r$ yields an update multiset, then the rule $r$ also yields an
%update set.

\begin{description}

  \item[\textbf{A1}] $\mathrm{conUSet}(X) \wedge \forall z (\neg X(f,x,z)) \wedge f(x)=y \rightarrow [X]f(x) = y$  \smallskip

  \item[\textbf{A2}] $\mathrm{conUSet}(X) \wedge X(f,x,y) \rightarrow [X]f(x)=y$\smallskip

%  \item[\textbf{\emph{A3}}] $\mathrm{upm}(r,\ddot{X}) \rightarrow \exists  X (\mathrm{upd}(r,X))$  \smallskip

\end{description}

\item The following are axiom schemes from classical logic.

\begin{description}
\item[\textbf{P1}] $\varphi\rightarrow(\psi\rightarrow\varphi)$\smallskip

\item[\textbf{P2}] $(\varphi\rightarrow(\psi\rightarrow\chi))\rightarrow((\varphi\rightarrow\psi)\rightarrow (\varphi\rightarrow\chi))$\smallskip

\item[\textbf{P3}]
$(\neg\varphi\rightarrow\neg\psi)\rightarrow(\psi\rightarrow\varphi)$\smallskip

\end{description}

\item The following four inference rules describe when the universal and existential quantifiers can be added to or deleted from a statement.
Rules~\textbf{UI}, \textbf{EG}, \textbf{UG} and~\textbf{EI} are usually known as \emph{Universal Instantiation}, \emph{Existential Generalisation}, \emph{Universal Generalisation} and \emph{Existential Instantiation}, respectively.

\begin{description}

\item[\textbf{\emph{UI}}] $\forall x (\varphi) \vdash \varphi[t/x]$ if $\varphi$ is pure or $t$ is static. \smallskip

\item[\textbf{\emph{EG}}] $\varphi[t/x] \vdash \exists x(\varphi)$ if $\varphi$ is pure or $t$ is static.\smallskip

\item[\textbf{\emph{UG}}] $\varphi[t_a/x] \vdash \forall x(\varphi)$ if $\varphi[t_a/x]$ holds for every element $a$ in the domain of $x$ and corresponding term $t_a$ representing $a$, and further $\varphi$ is pure or every $t_a$ is static.\smallskip

\item[\textbf{\emph{EI}}] $\exists x (\varphi) \vdash \varphi[t/x]$ if $t$ represents a valuation for $x$ which satisfies $\varphi$, and further $\varphi$ is pure or $t$ is static. \smallskip
\end{description}

\item  The following are the equality axioms from first-order logic with equality.
Axiom~\textbf{EQ1} asserts the reflexivity property while
Axiom~\textbf{EQ2} asserts the substitutions for functions.
%and Axiom~\textbf{EQ3} asserts the substitutions for $\rho$-terms.
%Again, terms occurring in the axioms are restricted to be static, which do
%not contain any dynamic function symbols.

\begin{description}

\item[\textbf{\emph{EQ1}}] $t=t$ for static term $t$ \smallskip

\item[\textbf{\emph{EQ2}}] $t_1=t_{n+1}\wedge...\wedge t_n=t_{2n}\rightarrow f(t_1,...,t_n) = f(t_{n+1},...,t_{2n})$ for any
function $f$ and static terms $t_i$ $(i=1,...,2n)$.\smallskip

%\item[\textbf{\emph{EQ3}}] $t_1=t_2\wedge(\varphi_1(x,\bar{y})\leftrightarrow\varphi_2(x,\bar{y}))\rightarrow \rho_{x}(t_1|\varphi_1(x,\bar{y}))=\rho_{x}(t_2|\varphi_2(x,\bar{y}))$ for pure formulae $\varphi_1$ and $\varphi_2$, and static terms $t_1$ and $t_2$. \smallskip

\end{description}

\smallskip

\item The following axiom is taken from dynamic logic, asserting that executing a \textbf{seq} rule equals to executing rules sequentially.

\begin{description}

%\item[\textbf{\emph{DY1}}] \ $\exists X. \text{upd}(\textbf{seq}\; r_1\text{ }r_2 \;\textbf{endseq}, X) \wedge [X]\varphi \leftrightarrow \exists X_1 . \text{upd}(r_1,X_1) \wedge \exists X_2 . $ $\text{upd}(r_2,X_2) \wedge [X_1] [X_2] \varphi$
\item[\textbf{\emph{DY1}}] \ $\exists X (\text{upd}(\textbf{seq}\; r_1\text{ }r_2 \;\textbf{endseq}, X) \wedge [X]\varphi) \leftrightarrow$\\
\hspace*{3.4cm} $\exists X_1 (\text{upd}(r_1,X_1) \wedge [X_1]\exists X_2(\text{upd}(r_2,X_2) \wedge [X_2] \varphi))$
\end{description}

\item Axiom \textbf{E} is the extensionality axiom.

\begin{description}
 \item[\textbf{\emph{E}}] $r_1\equiv r_2\rightarrow \exists X_1 X_2(($upd$(r_1,X_1)\wedge [X_1]\varphi)\leftrightarrow($upd$(r_2,X_2)\wedge [X_2]\varphi))$

\end{description}
\smallskip
\end{itemize}

The following soundness theorem for the proof system is relatively straightforward, since the non-standard axioms and rules are just a formalisation of the definitions of the semantics of rules, update sets and update multisets.

\begin{theorem}\label{c5-theoremsoundness}

Let $\varphi$ be a formula from ${\cal L}$ and let $\Phi$ be a set of formulae also from ${\cal L}$ (all of them of the same vocabulary as $\varphi$). If $\Phi\vdash\varphi$, then $\Phi\models\varphi$.

\end{theorem}

%%%%%%%%%%%%%%%%%%%%%%%%%%%%%%%%%%%%%%%%%%%%%%%%%%%%%%%%%%%%%%%%%%%%%%%%%%%%%%%%%%%%%%%%%%%%%%%%%%%%%%%%

\section{Derivation}\label{sec:Derivation}

In this section we present some properties of the logic for non-deterministic parallel ASMs which are implied by the axioms and rules from the previous section. This includes properties known for the logic for ASMs \cite{RobertLogicASM}. In particular, the logic for ASMs uses the modal expressions $[r]\varphi$ and
$\langle r\rangle\varphi$ with the following semantics:

\begin{itemize}

\item $[\![[r]\varphi]\!]_{S,\zeta}=\textit{true}$ iff $[\![\varphi]\!]_{S+\Delta,\zeta}=\mathit{true}$ for all consistent $\Delta\in \Delta(r,S,\zeta)$.

\item $[\![\langle r\rangle\varphi]\!]_{S,\zeta}=\mathit{true}$ iff $[\![\varphi]\!]_{S+\Delta,\zeta}=\mathit{true}$ for at least one consistent $\Delta\in \Delta(r,S,\zeta)$.

\end{itemize}

Instead of introducing modal operators $[\hspace{0.1cm}]$ and
$\langle\hspace{0.1cm}\rangle$ for a non-deterministic parallel
ASM rule $r$, we use the modal expression $[X]\varphi$ for an update set yielded by a
possibly non-deterministic rule. The modal expressions $[r]\varphi$
and $\langle r\rangle\varphi$ in the logic for ASMs can be treated
as the shortcuts for the following formulae in our logic: 
\begin{equation}\label{ASM1}
[r]\varphi \equiv \forall X (\text{upd}(r,X)\rightarrow[X]\varphi).
\end{equation}
\begin{equation}\label{ASM2}
\langle r\rangle\varphi \equiv\exists
X (\text{upd}(r,X)\wedge[X]\varphi).
\end{equation}

%$[r]\varphi \equiv \forall X (\text{upd}(r,X)\rightarrow[X]\varphi)$ and $\langle r\rangle\varphi \equiv\exists X (\text{upd}(r,X)\wedge[X]\varphi)$.

\begin{lemma}\label{lem-modalAxiomsASMs}
The following axioms and rules used in the logic for ASMs are derivable in $\cal{L}$, where the rule $r$ in Axioms (c) and (d) is assumed to be defined and deterministic: $(a)$ $([r](\varphi\rightarrow\psi)\rightarrow [r]\varphi)\rightarrow[r]\psi$;
$(b)$ $\varphi\rightarrow[r]\varphi$; $(c)$ $\neg$wcon$(r)\rightarrow [r]\varphi$; $(d)$ $[r]\varphi\leftrightarrow \neg[r]\neg\varphi$.
\end{lemma}

\begin{proof} We prove each property in the following.

\begin{itemize}
  \item (a): By Equation~\ref{ASM1}, we have that $[r](\varphi\rightarrow\psi)\wedge [r]\varphi \equiv \forall X (\mathrm{upd}(r,X)\rightarrow [X](\varphi\rightarrow\psi)) \wedge \forall X (\mathrm{upd}(r,X)\rightarrow [X]\varphi)$. By the axioms from classical logic, this is in turn equivalent to $\forall X (\mathrm{upd}(r,X)\rightarrow([X](\varphi\rightarrow\psi)\wedge[X]\varphi))$.
Then by Axiom~\textbf{M1} and axioms from the classical logic, we get $\forall X (\mathrm{upd}(r,X)\rightarrow ([X](\varphi\rightarrow\psi)\wedge [X]\varphi))\rightarrow \forall X (\mathrm{upd}(r,X)\rightarrow [X]\psi)$.
Therefore, $([r](\varphi\rightarrow\psi)\rightarrow [r]\varphi)\rightarrow [r]\psi$ is derivable.

%\begin{itemize}
%  \item By $[r]\varphi= \forall X. ($upd$(r,X)\rightarrow [X]\varphi)$, we have
%
%   $[r](\varphi\rightarrow\psi)\wedge [r]\varphi$$= \forall X. ($upd$(r,X)\rightarrow [X](\varphi\rightarrow\psi))\wedge\forall X. ($upd$(r,X)\rightarrow [X]\varphi)$.
%  \item By the axioms from classical logic, we have $[r](\varphi\rightarrow\psi)\wedge [r]\varphi$$= \forall X. ($upd$(r,X)\rightarrow ([X](\varphi\rightarrow\psi)\wedge [X]\varphi))$.
%  \item   Then by Axiom \textbf{M1}: $[X](\varphi\rightarrow\psi)\rightarrow [X]\varphi\rightarrow[X]\psi$, we can get
%  $\forall X. ($upd$(r,X)\rightarrow ([X](\varphi\rightarrow\psi)\wedge [X]\varphi))$$\rightarrow \forall X. ($upd$(r,X)\rightarrow [X]\psi)$.
%\end{itemize}

  \item (b): By Rule~\textbf{M2}, we have that $\varphi\rightarrow[X_i]\varphi$. Since $X$ is free in $\varphi\rightarrow[X]\varphi$, this holds for every possible valuation of $X$. Thus using Rule~\textbf{UG} (applied to the variable $X$ of the first predicate sort) and the axioms from classical logic, we can clearly derive $\varphi \rightarrow \forall X (\mathrm{upd}(r,X) \rightarrow [X]\varphi)$.

%As each DB-ASM rule is defined, we assume that
%  $\Delta(r,S,\zeta)=\{\Delta_1,...,\Delta_n\}$. By Rule \textbf{M2}, we have $\varphi\rightarrow[X_i]\varphi$ $(i=1,...,n)$ for all update sets $\{\Delta_1,..,\Delta_n\}$ generated by $r$, i.e., $\%varphi\rightarrow\forall X. ($upd$(r,X)\rightarrow
%  [X]\varphi)$. By Equation \ref{ASM1}, we get
%  $\varphi\rightarrow[r]\varphi$. Therefore, $\varphi\rightarrow[r]\varphi$ is derivable.

  \item (c): By Equation~\ref{wcon}, we have $\neg \mathrm{wcon}(r)\leftrightarrow \neg\exists
  X (\mathrm{con}(r,X)$. In turn, by Equation~\ref{conr}, we get $\neg \mathrm{wcon}(r) \leftrightarrow \neg \exists X (\mathrm{upd}(r,X)\wedge \mathrm{conUSet}(X))$.
  Since a rule $r$ in the logic for ASMs is deterministic, we get $\neg \mathrm{wcon}(r)\leftrightarrow \neg \mathrm{conUSet}(X)$. By Axiom \textbf{M4}, we get $\neg \mathrm{wcon}(r)\rightarrow
  [r]\varphi$.

  \item (d): By Equation~\ref{ASM1}, we have $\neg[r]\neg\varphi \equiv \exists X(\mathrm{upd}(r,X)\wedge \neg[X]\neg\varphi)$. By applying Axiom~\textbf{M5} to $\neg[X]\neg\varphi$, we get
$\neg[r]\neg\varphi \equiv\exists X(\mathrm{upd}(r,X)\wedge[X]\varphi)$. When the rule $r$ is deterministic, the interpretation of $\forall X (\mathrm{upd}(r,X)\rightarrow [X]\varphi)$ coincides with he interpretation of $\exists X (\mathrm{upd}(r,X)\wedge [X]\varphi)$ and therefore $[r]\varphi\leftrightarrow \neg[r]\neg\varphi$.
\end{itemize}
\end{proof}
%The logic for ASMs introduced in \cite{RobertLogicASM} is deterministic, i.e., it excludes nondeterministic choice rules.
%In contrast, our logic for DB-ASMs includes a nondeterministic choice rule.

Note that the formula Con$(R)$ in \textbf{Axiom 5} in \cite{RobertLogicASM} (i.e.,
in $\neg$Con$(R)\rightarrow [R]\varphi$) corresponds to the weak version of consistency (i.e., wcon$(r)$)
in the theory of ${\cal L}$.

\begin{lemma}\label{lem-soundness-modaloperator}
The following properties are derivable in $\cal{L}$: $(e)$ $\mathrm{con}(r,X) \wedge [X]f(x)=y\rightarrow X(f,x,y)\vee(\forall z (\neg X(f,x,z)) \wedge f(x)=y)$; $(f)$ $\mathrm{con}(r,X) \wedge [X]\varphi \rightarrow \neg[X]\neg\varphi$; $(g)$ $[X]\exists x(\varphi) \rightarrow \exists x ([X]\varphi)$; $(h)$ $[X]\varphi_1 \wedge [X]\varphi_2 \rightarrow [X] (\varphi_1 \wedge \varphi_2)$.
\end{lemma}

\begin{proof}
(e) is derivable by applying Axioms~\textbf{A1} and~\textbf{A2}. (f)
is a straightforward result of Axiom~\textbf{M5}. (g) can be derived
by applying Axioms~\textbf{M5} and~\textbf{M6}. Regarding (h), it is
derivable by using Axioms \textbf{M1}-\textbf{M3}.
\end{proof}

\begin{lemma}For terms and variables of the appropriate types, the following properties in \cite{GroenboomFLEA95} are derivable in $\cal{L}$.
\begin{itemize}
  \item $x=t \rightarrow (y=s \leftrightarrow [f(t):=s] f(x)=y)$
  \item $x \neq t \rightarrow (y = f(x) \leftrightarrow[f(t):=s]f(x)=y)$
\end{itemize}
\end{lemma}

%In non-deterministic parallel ASMs, two parallel computations may produce an update multiset,
%in which there are identical updates to a location assigned with a
%location operator. Without an outer \textbf{let}-construct \textbf{par} $ r \; r $ \textbf{endpar}
%could be simplified to $r$. This however is no longer the case if we consider update multisets.

Following the approach of defining the predicate joinable in
\cite{RobertLogicASM}, we define the predicate joinable over two
non-deterministic parallel ASMs rules. As we consider non-deterministic parallel ASMs rules,
the predicate joinable$(r_1,r_2)$ means that there exists a pair of
update sets without conflicting updates, which are yielded by rules
$r_1$ and $r_2$, respectively. Then, based on the use of predicate
joinable, the properties in Lemma \ref{lem-soundness-consistency}
are all derivable.
\begin{equation}\label{joinable}
\begin{split}
\text{joinable}(r_1,r_2) \equiv &  \exists X_1 X_2 (\mathrm{upd}(r_1,X_1)\wedge\mathrm{upd}(r_2,X_2)\wedge \\
  &  \bigwedge\limits_{f\in\mathcal{F}_{dyn}}\forall x y z (X_1(f,x,y)\wedge X_2(f,x,z)\rightarrow y=z ))
\end{split}
\end{equation}

\begin{lemma}\label{lem-soundness-consistency}
The following properties for weak consistency are derivable in ${\cal L}$.
\begin{description}
  \item[(i)] $\mathrm{wcon}(f(t):=s) \qquad \mathrm{{\bf (j)}} \,\mathrm{wcon}(f(t):=\mathtt{s}) \qquad \mathrm{{\bf (k)}} \,\mathrm{wcon}(f(\mathtt{t}):=\mathtt{\mathtt{s}})$
  \item[(j)] $\mathrm{wcon}(\textbf{if}\, \varphi\, \textbf{then}\, r \,\textbf{endif}) \leftrightarrow \neg\varphi\vee(\varphi\wedge \mathrm{wcon}(r))$
  \item[(l)] $\mathrm{wcon}(\textbf{forall} \, x \, \textbf{with} \, \varphi \, \textbf{do} \, r \, \textbf{enddo}) \leftrightarrow$\\
    \hspace*{3.5cm}$\forall x (\varphi\rightarrow \mathrm{wcon}(r) \wedge \forall y (\varphi[y/x]\rightarrow \text{joinable}(r,r[y/x])))$
\item[(m)] $\mathrm{wcon}(\textbf{par} \, r_1 \, r_2 \, \textbf{endpar}) \leftrightarrow \mathrm{wcon}(r_1) \wedge \mathrm{wcon}(r_2) \wedge joinable(r_1,r_2)$
  \item[(n)] $\mathrm{wcon}(\textbf{choose} \, x \, \textbf{with} \, \varphi \, \textbf{do} \, r \, \textbf{enddo}) \leftrightarrow \exists x (\varphi\wedge \mathrm{wcon}(r))$
  \item[(o)] $\mathrm{wcon}(\textbf{choose} \, \mathtt{x} \, \textbf{with} \, \varphi \, \textbf{do} \, r \, \textbf{enddo}) \leftrightarrow \exists \mathtt{x} (\varphi\wedge \mathrm{wcon}(r))$
  \item[(p)] $\mathrm{wcon}(\textbf{seq} \, r_1 \, r_2 \, \textbf{endseq}) \leftrightarrow \exists X (\mathrm{con}(r_1,X) \wedge [X]\mathrm{wcon}(r_2))$
%  \item[(o)] $\mathrm{wcon}(\textbf{let} \, (f,t) \!\rightharpoonup\! \rho \, \textbf{in} \, r \textbf{endlet}) \leftrightarrow$\\
%\hspace*{1.3cm}$\exists X Y (\mathrm{upd}(r,X) \wedge \mathrm{conUSet}(Y) \wedge \forall x y (Y(f,x,y) \leftrightarrow (t = x \vee X(f,x,y))) \wedge $\\
%\hspace*{2cm}   $\bigwedge\limits_{f \neq f'\in\mathcal{F}_{dyn}}\forall x y (X(f',x,y) \leftrightarrow Y(f',x,y)))$
\end{description}
\end{lemma}

We omit the proof of the previous lemma as well as the proof of the remaining lemmas in this section, since they are lengthy but relatively easy exercises.

\begin{lemma}\label{lem-soundness-rule}The following properties for the formula $[r]\varphi$ are derivable in $\cal{L}$.
\begin{description}
  \item[(q)] $[\textbf{if} , \varphi , \textbf{then} , r , \textbf{endif}]\psi \leftrightarrow (\varphi\wedge[r]\psi) \vee (\neg \varphi \wedge \psi)$
  \item[(r)] $[\textbf{choose} \, x \, \textbf{with} \, \varphi \, \textbf{do} \, r \, \textbf{enddo}]\psi \leftrightarrow \forall x (\varphi\rightarrow [r]\psi)$
  \item[(s)] $[\textbf{choose} \, \mathtt{x} \, \textbf{with} \, \varphi \, \textbf{do} \, r \, \textbf{enddo}]\psi \leftrightarrow \forall \mathtt{x} (\varphi\rightarrow [r]\psi)$\smallskip
\end{description}
\end{lemma}

Lemma \ref{lem-soundness-composition} states that a parallel
composition is commutative and associative while a sequential
composition is associative.

\begin{lemma}\label{lem-soundness-composition}The following properties are derivable in $\cal{L}$.
\begin{description}
  \item[(t)] \textbf{par} $r_1\hspace{0.2cm} r_2$ \textbf{endpar} $\equiv$ \textbf{par} $r_2\hspace{0.2cm} r_1$ \textbf{endpar}
  \item[(u)] \textbf{par} (\textbf{par} $r_1\hspace{0.2cm} r_2$ \textbf{endpar}) $r_3$ \textbf{endpar} $\equiv$ \textbf{par} $r_1$ (\textbf{par} $r_2\hspace{0.2cm} r_3$
  \textbf{endpar}) \textbf{endpar}
  \item[(v)] \textbf{seq} (\textbf{seq} $r_1\hspace{0.2cm} r_2$ \textbf{endseq}) $r_3$ \textbf{endseq} $\equiv$ \textbf{seq} $r_1$ (\textbf{seq} $r_2\hspace{0.2cm} r_3$
  \textbf{endseq}) \textbf{endseq}
\end{description}
\end{lemma}

\begin{lemma}The extensionality axiom for transition rules in the logic for ASMs is derivable in $\cal{L}$: $r_1\equiv r_2\rightarrow([r_1]\varphi\leftrightarrow [r_2]\varphi)$.
\end{lemma}

%%%%%%%%%%%%%%%%%%%%%%%%%%%%%%%%%%%%%%%%%%%%%%%%%%%%%%%%%%%%%%%%%%%%%%%%%%%%%%%%%%%%%%%%%%%%%%%%%%%%%%%%%%%%%%%%%%%%%%%%%%%
\section{Completeness}\label{sec:completeness}

%In this section we investigate the completeness of the proof system.
%Since the non-deterministic parallel ASMs are a variation of hierarchical ASMs, which do not have
%recursive rule declarations, 
We can prove the completeness of ${\cal L}$ by using a similar strategy to that used in~\cite{RobertLogicASM}. That is, we can show that $\cal L$ is a definitional extension of a complete logic. However, the logic for hierarchical ASMs in~\cite{RobertLogicASM} is a definitional extension of first-order logic. In the case of the logic $\cal L$, the proof is more complicated since we have to deal with set membership predicates and corresponding predicate sorts. The key idea is to show instead that ${\cal L}$ is a
\emph{definitional extension} of first-order logic extended with two membership predicates with respect to finite sets, which in turns constitutes itself a complete logic.

In the remaining of this section, we will use ${\cal L}^\in$ to denote the logic obtained by restricting the formulae of ${\cal L}$ to those produced by the following grammar:\\[0.1cm]
$\varphi, \psi \, ::= \, s=t \mid s_a = t_a \mid  \neg\varphi \mid \varphi \wedge \psi \mid \forall x (\varphi) \mid \forall \texttt{x} (\varphi) \mid \forall x^1 (\varphi) \mid \forall x^2 (\varphi) \mid $\\[0.1cm]
\hspace*{1.4cm}$\in^1\!\!(x^1\!,f,t_0,s_0) \mid \in^2\!\!(x^2\!,f,t_0,s_0,s).$\\

Let us define the theory of ${\cal L}^\in$ as the theory obtained by taking the union of a sound and complete axiomatisation of first-order logic and the sound and complete axiomatisation of the properties of finite sets introduced in \cite{agotness:ljigpl2008}. Clearly, such theory of ${\cal L}^\in$ is a conservative extension of the first-order theory, in the sense that if $\Phi$ is a set of pure first-order formulae and $\varphi$ is a pure first-order formula (not containing subformulae of the form $\in^n\!\!(x^n, t_1, \ldots, t_n)$) and $\Phi \vdash \varphi$ holds in the theory of ${\cal L}^\in$, then there already exists a derivation using the axiomatisation for first-order logic. Indeed, due to the soundness of the axioms and rules in the theory of ${\cal L}^\in$, we obtain $\Phi \models \varphi$, which is a pure statement about models for first-order logic. Thus the known completeness for first-order logic gives $\Phi \vdash \varphi$ in an axiomatisation for first-order logic, hence the claimed conservativism of the extension. Since then the theory of ${\cal L}^\in$ proves no new theorems about first-order logic, all the new theorems belong to the theory of properties of finite sets and thus can be derived by using the axiomatisation in \cite{agotness:ljigpl2008} (which also form part of the axiomatisation of ${\cal L}^\in$), we get the following key result.

\begin{theorem}\label{CompletenessLin}
Let $\varphi$ be a formula and $\Phi$ be a set of formulae in the language of ${\cal L}^\in$ (all of the same vocabulary). If $\Phi\models\varphi$, then $\Phi\vdash\varphi$.
\end{theorem}

Finally, we need to show that all the formulae in ${\cal L}$ which are not formulae of ${\cal L}^\in$ can be translated into formulae of ${\cal L}^\in$ based on derivable equivalences in the theory of ${\cal L}$.
First, we reduce the general atomic formulae in ${\cal L}$ to atomic formulae of the form $x=y$, $\mathtt{x} = \mathtt{y}$, $f(x)=y$, $f(x)=\mathtt{y}$, $f(\mathtt{x})=\mathtt{y}$, $\in^1\!\!(x^1\!,f,x,y)$, $\in^1\!\!(x^1\!,f,x,\mathtt{y})$, $\in^1\!\!(x^1\!,f,\mathtt{x},\mathtt{y})$, $\in^2\!\!(x^2\!,f,x,y, z)$, $\in^2\!\!(x^2\!,f,x,\mathtt{y},z)$ and $\in^2\!\!(x^2\!,f,\mathtt{x},\mathtt{y}, z)$.  Let $t$, $s$ and $s'$ denote point terms and let $t_a$ and $s_a$ denote algorithmic terms. This can be done by using the following equivalences.
\begin{align*}
s=t &\leftrightarrow \exists x (s = x \wedge x = t)\\
s_a = t_a &\leftrightarrow \exists \mathtt{x} (s_a = \mathtt{x} \wedge \mathtt{x} = t_a)\\
f(s)=y &\leftrightarrow \exists x (s = x \wedge f(x) = y)\\
f(s)=\mathtt{y} &\leftrightarrow \exists x (s = x \wedge f(x) = \mathtt{y})\\
f(s_a)=\mathtt{y} &\leftrightarrow \exists \mathtt{x} (s_a = \mathtt{x} \wedge f(\mathtt{x}) = \mathtt{y})\\
\in^1\!\!(x^1\!,f,t,s) &\leftrightarrow \exists x y (t = x \wedge s = y \wedge \in^1\!\!(x^1\!,f,x,y))\\
\in^1\!\!(x^1\!,f,t,s_a) &\leftrightarrow \exists x \mathtt{y} (t = x \wedge s_a = \mathtt{y} \wedge \in^1\!\!(x^1\!,f,x,\mathtt{y}))\\
\in^1\!\!(x^1\!,f,t_a,s_a) &\leftrightarrow \exists \mathtt{x} \mathtt{y} (t_a = \mathtt{x} \wedge s_a = \mathtt{y} \wedge \in^1\!\!(x^1\!,f,\mathtt{x},\mathtt{y}))\\
\in^2\!\!(x^2\!,f,t,s,s') &\leftrightarrow \exists x y z(t = x \wedge s = y \wedge s' = z \wedge \in^2\!\!(x^2\!,f,x,y,z))\\
\in^2\!\!(x^2\!,f,t,s_a,s') &\leftrightarrow \exists x \mathtt{y} z (t = x \wedge s_a = \mathtt{y} \wedge s' = z \wedge \in^2\!\!(x^2\!,f,x,\mathtt{y},z))\\
\in^2\!\!(x^2\!,f,t_a,s_a,s') &\leftrightarrow \exists \mathtt{x} \mathtt{y} z (t_a = \mathtt{x} \wedge s_a = \mathtt{y} \wedge s' = z \wedge \in^2\!\!(x^2\!,f,\mathtt{x},\mathtt{y},z))
\end{align*}
The translation of modal formulae into ${\cal L}^\in$ distributes over negation, Boolean connectives and quantifiers. We eliminate atomic formulae of the form $\mathrm{upd}(r, x^1)$ using Axioms~\textbf{U1}-\textbf{U7}, and the modal operator in formulae of the form $[x^1] \varphi$, where $\varphi$ is already translated to ${\cal L}^\in$, using the following derivable equivalences. \\[0.2cm]
$[x^1]x=y \leftrightarrow(\text{conUSet}(x^1)\rightarrow x=y)$; \hfill $[x^1]\mathtt{x}=\mathtt{y} \leftrightarrow(\text{conUSet}(x^1)\rightarrow \mathtt{x}=\mathtt{y})$; \\
$[x^1]f(x)=y \leftrightarrow (\text{conUSet}(x^1)\rightarrow \, \in^1\!\!(x^1\!, f, x, y)\vee(\forall z (\neg\! \in^1\!\!(x^1\!, f, x, z)) \wedge f(x) = y))$\\
$[x^1]f(x)=\mathtt{y} \leftrightarrow (\text{conUSet}(x^1)\rightarrow \, \in^1\!\!(x^1\!, f, x, \mathtt{y})\vee(\forall \mathtt{z} (\neg\! \in^1\!\!(x^1\!, f, x, \mathtt{z})) \wedge f(x) = \mathtt{y}))$\\
$[x^1]f(\mathtt{x})=\mathtt{y} \leftrightarrow (\text{conUSet}(x^1)\rightarrow \, \in^1\!\!(x^1\!, f, \mathtt{x}, \mathtt{y})\vee(\forall \mathtt{z} (\neg\! \in^1\!\!(x^1\!, f, \mathtt{x}, \mathtt{z})) \wedge f(\mathtt{x}) = \mathtt{y}))$;\\
$[x^1]\!\! \in^1\!\!(x^1\!, f, x, y) \leftrightarrow(\text{conUSet}(x^1)\rightarrow \, \in^1\!\!(x^1\!, f, x, y))$; \\
$[x^1]\!\! \in^1\!\!(x^1\!, f, x, \mathtt{y}) \leftrightarrow(\text{conUSet}(x^1)\rightarrow \, \in^1\!\!(x^1\!, f, x, \mathtt{y}))$; \\
$[x^1]\!\! \in^1\!\!(x^1\!, f, \mathtt{x}, \mathtt{y}) \leftrightarrow(\text{conUSet}(x^1)\rightarrow \, \in^1\!\!(x^1\!, f, \mathtt{x}, \mathtt{y}))$; \\
$[x^1]\!\! \in^2\!\!(x^2\!, f, x, y, z) \leftrightarrow(\text{conUSet}(x^1)\rightarrow \, \in^2\!\!(x^2\!, f, x, y, z))$; \\
$[x^1]\!\! \in^2\!\!(x^2\!, f, x, \mathtt{y}, z) \leftrightarrow(\text{conUSet}(x^1)\rightarrow \, \in^2\!\!(x^2\!, f, x, \mathtt{y},z))$; \\
$[x^1]\!\! \in^2\!\!(x^2\!, f, \mathtt{x}, \mathtt{y}, z) \leftrightarrow(\text{conUSet}(x^1)\rightarrow \, \in^2\!\!(x^2\!, f, \mathtt{x}, \mathtt{y},z))$; \\
$[x^1]\neg\varphi \leftrightarrow(\text{conUSet}(x^1)\rightarrow \neg[x^1]\varphi)$; \hspace{1.175cm} $[x^1](\varphi\wedge\psi) \leftrightarrow([x^1]\varphi\wedge[x^1]\psi)$; \\
$[x^1]\forall x (\varphi) \leftrightarrow \forall x ([x^1]\varphi)$; \hspace{3.03cm} $[x^1]\forall \mathtt{x} (\varphi) \leftrightarrow \forall \mathtt{x} ([x^1]\varphi)$; \\
$[x^1]\forall y^1 (\varphi) \leftrightarrow \forall y^1 ([x^1]\varphi)$; \hspace{2.75cm} $[x^1]\forall x^2 (\varphi) \leftrightarrow \forall x^2 ([x^1]\varphi)$.\\

Our main technical result then follows from Theorem~\ref{CompletenessLin} and the fact that the described translation from formulae $\varphi$ of ${\cal L}$ to formulae $\varphi'$ of ${\cal L}^\in$ satisfies the properties required for ${\cal L}$ to be a definitional extension of ${\cal L}^\in$, i.e., (a) $\varphi\leftrightarrow \varphi'$ is derivable in ${\cal L}$ and (b) $\varphi'$ is derivable in ${\cal L}^\in$ whenever $\varphi$ is derivable ${\cal L}$. 

\begin{theorem} \label{thm-adtm-completeness-FOL}

Let $\varphi$ be a formula and $\Phi$ a set of formulae in the language of ${\cal L}$ (all of the same vocabulary). If $\Phi\models\varphi$, then $\Phi\vdash\varphi$.

\end{theorem}

\section{Conclusion}\label{conclusions}

Non-deterministic transitions manifest themselves as a difficult task in the logical formalisation for ASMs. Indeed, Nanchen and St\"ark analysed potential problems to several approaches they tried by taking non-determinism into consideration and concluded \cite{RobertLogicASM}:
\begin{quote}
Unfortunately, the formalisation of consistency cannot be applied directly to non-deterministic ASMs. The formula Con$(r)$ (as defined in Sect. 8.1.2 of \cite{boerger:2003}) expresses the property that the \emph{union of all possible} update sets of (an ASM rule) $r$ in a given state is consistent. This is clearly not what is meant by consistency. Therefore, in a logic for ASMs with \textbf{choose} one had to add Con$(r)$ as an atomic formula to the logic.
\end{quote}

However, we observe that this conclusion is not necessarily true, as finite update sets can be made explicit in the formulae of a logic to capture non-deterministic transitions. In doing so, the formalisation of consistency defined in \cite{RobertLogicASM} can still be applied to such an explicitly specified update set $\Delta$ yielded by a rule $r$ in the form of the formula $\mathrm{con}(r,\Delta)$ as discussed in Subsection~\ref{sub:Consistency}. We thus solve this problem by the addition of the modal operator $[\Delta]$ for an update set generated by a non-deterministic parallel ASM rule. The approach works well, because in the parallel ASMs the number of possible parallel branches, although unbounded, is still finite. Therefore the update sets produced by these machines are restricted to be finite as well. This is implicitly assumed in the parallel ASM thesis of Blass and Gurevich\cite{blass:tocl2003,GurevichParallelCorrection08} and it is made explicit in the new parallel ASM thesis that we propose in \cite{FerrarottiSTW15}.

The proof systems that we develop in this work for the proposed logic for non-deterministic parallel ASMs, extends the proof system developed in \cite{RobertLogicASM} in two different ways. First, an ASM rule may be associated with a set of different update sets. Applying different update sets may lead to a set of different successor states to the current state. As the logic for non-deterministic parallel ASMs includes formulae denoting explicit update sets and variables that are bounded to update sets, our proof system allows us to reason about the interpretation of a formula over all successor states or over some successor state after applying an ASM rule over the current state. Secondly, in addition to capturing the consistency of an update set yielded by an ASM rule, our proof system also develops two notions of consistency (weak and strong consistency) w.r.t. a given rule. When the rule is deterministic, these two notions coincide.

%This article formalizes a one-step logic which captures statements abut the main rule of an ASM (where the ASM iterates the rule), investigates derivation rules for the logic, and proves soundness and completeness. The state of the art in this area was represented by the logic for ASMs introduced by Nanchen and St\"ark in~\cite{RobertLogicASM}. We go beyond this state of the art by incorporating the handling of bounded as well as unbounded non-deterministic ASM rules into our logical formalism.  
%Thus, we establish a sound and complete proof system for the logic for ASMs, which can be turned into a tool for reasoning about a considerably broader class of algorithms than previously. 

We plan as future work to embed our one-step logic into a complex dynamic logic and demonstrate how desirable properties of ASM runs can be formalised in such a logic. Of course, there is no chance of obtaining a complete proof theory for full ASM runs, but there is clearly  many potential practical benefits from the perspective of the ASM method for systems development~\cite{boerger:2003}.

\bibliographystyle{splncs03}

\bibliography{DBTsLogic}

\end{document}